\documentclass[conference]{IEEEtran}
\IEEEoverridecommandlockouts
\usepackage{cite}
\usepackage{amsmath,amssymb,amsfonts,amsthm}
\usepackage{algorithmic}
\usepackage{graphicx}
\usepackage{textcomp}
\usepackage{xcolor}
\def\BibTeX{{\rm B\kern-.05em{\sc i\kern-.025em b}\kern-.08em T\kern-.1667em\lower.7ex\hbox{E}\kern-.125emX}}

\newcommand{\tQ}{\widetilde{Q}}
\newcommand{\tD}{\widetilde{D}}
\newcommand{\cB}{\mathcal{B}}
\newcommand{\cE}{\mathcal{E}}
\newcommand{\cF}{\mathcal{F}}
\newcommand{\cH}{\mathcal{H}}
\newcommand{\cP}{\mathcal{P}}
\newcommand{\cS}{\mathcal{S}}
\newcommand{\bra}[1]{\langle {#1} |}
\newcommand{\ket}[1]{| {#1} \rangle}
\newcommand{\proj}[1]{\ensuremath{|#1\rangle\!\langle#1|}}

\DeclareMathOperator{\spec}{spec}

\newtheorem{theorem}{Theorem}
\newtheorem{corollary}[theorem]{Corollary}
\newtheorem{lemma}[theorem]{Lemma}
\newtheorem{proposition}[theorem]{Proposition}

\begin{document}

\title{Chain rules for quantum channels}

\author{\IEEEauthorblockN{Mario Berta}
\IEEEauthorblockA{\textit{Department of Computing, Imperial College London, UK} \\
\textit{AWS Center for Quantum Computing, Pasadena, USA} \\
\textit{California Institute of Technology, Pasadena, USA}}
\and
\IEEEauthorblockN{Marco Tomamichel}
\IEEEauthorblockA{\textit{Center for Quantum Technologies} \& \\
\textit{Department of Electrical and Computer Engineering}\\
\textit{National University of Singapore, Singapore}}
}

\maketitle


\begin{abstract}
Divergence chain rules for channels relate the divergence of a pair of channel inputs to the divergence of the corresponding channel outputs. An important special case of such a rule is the data-processing inequality, which tells us that if the same channel is applied to both inputs then the divergence cannot increase. 
Based on direct matrix analysis methods, we derive several R\'enyi divergence chain rules for channels in the quantum setting. Our results simplify and in some cases generalise previous derivations in the literature.
\end{abstract}





\section{Introduction}

A family of fundamental information measures is given by the R\'enyi divergences~\cite{renyi61}, defined for probabilit distributions $P,Q$ over finite alphabets $\mathcal{X}$ and $\alpha \in (1,\infty)$ as
\begin{align}\label{eq:renyi-classical}
D_{\alpha}(P\|Q)&=\frac{1}{\alpha-1} \log Q_{\alpha}(P\|Q)\quad\text{with}\nonumber\\
Q_{\alpha}(P\|Q)&=\sum_{x \in \mathcal{X}} P( x) \left( \frac{P(x)}{ Q(x)} \right)^{\alpha-1}\,,
\end{align}
if $P$ is absolutely continuous with respect to $Q$ and as $+\infty$ otherwise. This definition extends via the respective limits to $\alpha\to1$ (the Kulback-Leibler divergence) and $\alpha\to\infty$ (the max-divergence). They are respectively given as
\begin{align}
D(P\|Q) =\ &D_1(P\|Q) = \sum_{x \in \mathcal{X}} P(x) \log \frac{P(x)}{Q(x)} \quad \textrm{and} \\
&D_{\infty}(P\|Q) = \max_{x} \log \frac{P(x)}{Q(x)} \,.
\end{align}
For $\alpha \in (0,1)$ we rewrite the sum as $\sum_{x \in \mathcal{X}} P(x)^{\alpha} Q(x)^{1-\alpha}$, indicating that absolutely continuity of $P$ with respect to $Q$ is not necessary to keep $D_{\alpha}$ finite for $\alpha < 1$.

For distributions over bipartite alphabets $\mathcal{X}\times\mathcal{Y}$ the R\'enyi divergences enjoy the chain rule property
\begin{align}\label{eq:classical_chain}
&D_{\alpha}(P_{XY}\|Q_{XY})\nonumber\\
&\qquad \leq D_{\alpha}(P_X\|Q_X)+\max_{x\in\mathcal{X}}D_{\alpha}(P_{Y|X=x}\|Q_{Y|X=x}) ,
\end{align}
displaying how to split up R\'enyi divergences with respect to multipartite constituents. Eq.~\eqref{eq:classical_chain} is important in statistics because it can be seen as the crucial entropy inequality to quantify the asymptotic error exponents of adaptive channel discrimination \cite{hayashi09d,polyanski11,berta19,wilde20,salek20}.

In this report, we give various quantum generalizations of Eq.~\eqref{eq:classical_chain}, where probability distributions and conditional probability distributions are replaced with positive operators and positive maps, respectively. Some of these quantum chain rules have been derived previously in the context of adaptive quantum channel discrimination and related tasks \cite{christandl17,kun20,fang21,fawzi21}. The proofs of these earlier results are based on involved techniques from either complex interpolation theory \cite{christandl17}, smooth entropies \cite{kun20}, or convex optimization theory \cite{fawzi21}. In contrast, our derivations are based on direct matrix analysis methods.

The remainder of this report is structured as follows. First, we introduce in Section \ref{sec:methods} our notation and state some mathematical facts from quantum information theory. We then state our main result in Section \ref{sec:main}, after which we sketch applications to quantum channel discrimination in Section \ref{sec:applications}. We conclude with some open question in Section \ref{sec:outlook}. 


\section{Notation and methods}\label{sec:methods}

We consider quantum systems described by finite-dimensional inner product spaces labelled by $\cH$. The set $\cS(\cH)$ of quantum states is given by positive semi-definite trace one operators on $\cH$. For $\rho,\sigma\in\cS(\cH)$, we denote by $\rho\ll\sigma$ that the support of $\rho$ is contained in the support of $\sigma$, and inverses of operators are understood as generalized inverses on the support of the operator. Quantum channels are completely positive and trace preserving maps from $\cS(\cH)$ to $\cS(\cH')$. Bipartite quantum systems are described by tensor product spaces $\cH_A\otimes\cH_R$ with correspondingly labelled quantum states $\rho_{AR}\in\cS(\cH_A\otimes\cH_R)$.

An important tool to lift entropy inequalities from the classical to the quantum setting is given by the asymptotic spectral pinching method (see~\cite{hayashi02b}). Namely, the spectral pinching map $P_\sigma(\cdot)$ with respect to $\sigma\in\cS(\cH)$ is given as
\begin{align}
P_\sigma\;:\;\rho\mapsto\sum_{\lambda}P_\lambda\rho P_\lambda\,,
\end{align}
where $\sigma=\sum_\lambda \lambda P_\lambda$ with eigenvalues $\lambda\in\text{spec}(\sigma)$ in the spectrum of $\sigma$ and corresponding mutually orthogonal eigenprojectors $P_\lambda$. An alternative representation of the pinching map is
\begin{align}\label{eq:pinching-measure}
P_\sigma(\cdot)=\int\mu(\mathrm{d}t)\,\sigma^{it}(\cdot)\sigma^{-it}
\end{align}
for some probability measure $\mu$ on $\mathbb{R}$ \cite[Lemma 2.1]{sutter16}. We have that $P_\sigma(\rho)$ commutes with $\sigma$, $\text{Tr}[P_\sigma(\rho)\sigma]=\text{Tr}[\rho\sigma]$, and the pinching inequality \cite{hayashi02b}
\begin{align}\label{eq:pinching}
P_\sigma(\rho)\succeq\big|\text{spec}(\sigma)\big|^{-1}\rho\,,
\end{align}
where here and henceforth $\succeq$ denotes the L\"owner order, and $\big|\text{spec}(\sigma)\big|$ the number of distinct eigenvalues $\sigma$. Eq.~\eqref{eq:pinching} becomes asymptotically powerful as $\left|\text{spec}\left(\sigma^{\otimes n}\right)\right|= O(\text{poly}(n))$ by a type-counting argument~\cite{csiszar98}.

The generalization of the R\'enyi divergences from Eq.~\eqref{eq:renyi-classical} to the quantum setting is not unique. Rather, one asks for any quantum R\'enyi divergence
\begin{align}
\mathbb{D}_{\alpha}(\rho\|\sigma)=\frac{1}{\alpha-1}\log\mathbb{Q}_{\alpha}(\rho\|\sigma)\quad\text{on $\rho,\sigma\in\cS(\cH)$}
\end{align}
to satisfy the following properties:
\begin{enumerate}
\item whenever $\rho$ and $\sigma$ commute, $\mathbb{Q}_{\alpha}(\rho\|\sigma)$ simplifies to the corresponding classical functions in Eq.~\eqref{eq:renyi-classical} and, moreover, for classical quantum states $\rho = \sum_x p_x \proj{x} \otimes \rho_x$ and $\sigma = \sum_x q_x \proj{x} \otimes \sigma_x$, we have
\begin{align}
\mathbb{Q}_{\alpha}(\rho\|\sigma) = \sum_x p_x^{\alpha} q_x^{1-\alpha} \mathbb{Q}_{\alpha}(\rho_x\|\sigma_x) \,;
\end{align}
\item one has the data processing inequality
\begin{align}\label{eq:data-processing}
\mathbb{D}_{\alpha}(\rho\|\sigma)\geq\mathbb{D}_{\alpha}(\cE(\rho)\|\cE(\sigma))
\end{align}
for quantum channels $\cE:\cS(\cH)\to\cS(\cH')$.
\end{enumerate}
The smallest such quantum divergences for $\alpha>0$ are the measured R\'enyi divergence \cite{donald86,hiai91,Berta2017,Berta2016}
\begin{align}
D_{\alpha}^M(\rho\|\sigma)=\sup_{\mathcal{M}}D_{\alpha}(\mathcal{M}(\rho)\|\mathcal{M}(\sigma))
\end{align}
for the rank-one projective measurement maps $\mathcal{M}(\cdot)=\sum_x\bra{x}\cdot\ket{x}\proj{x}$ with the supremum over orthonormal bases $\{\ket{x}\}_x$ of $\cH$.\footnote{Including general measurement maps leads to the same quantities \cite[Theorem 4]{Berta2017}.} On the other hand, for $\alpha\in(0,2]$ the largest such quantum divergences are the geometric R\'enyi divergences \cite{matsumoto14,fang21}
\begin{align}\label{eq:geometric}
\widehat{D}_{\alpha}(\rho\|\sigma)&=\frac{1}{\alpha-1}\log\widehat{Q}_{\alpha}(\rho\|\sigma)\quad\text{with}\nonumber\\
\widehat{Q}_{\alpha}(\rho\|\sigma)&=\text{Tr}\left[\sigma^{1/2}\left(\sigma^{-1/2}\rho\sigma^{-1/2}\right)^{\alpha}\sigma^{1/2}\right]\;\text{for $\rho\ll\sigma$}
\end{align}
and the respective limit for $\alpha\to1$ as \cite{belavkin82,matsumoto10}
\begin{align}
\widehat{D}_1(\rho\|\sigma)=\mathrm{Tr}\left[\rho\log\left(\rho^{1/2}\sigma^{-1}\rho^{1/2}\right)\right]\,.
\end{align}
The formal definition from Eq.~\eqref{eq:geometric} can be extended to $\alpha>0$, but only for $\alpha\in(0,2]$ the geometric R\'enyi divergences have the following operational characterisation.

\begin{lemma}[Matsumoto \cite{matsumoto10,matsumoto14}]\label{lem:matsumoto}
Let $\rho,\sigma\in\cS(\cH)$ and $\alpha\in(0,2]$. Then, we have
\begin{align}
\widehat{D}_{\alpha}(\rho\|\sigma)=\inf_{(\Gamma,P,Q)}D_{\alpha}(P\|Q)\,,
\end{align}
where the infimum is over probability distributions $P,Q$ over finite alphabets $\mathcal{X}$, and positive trace preserving maps with $\Gamma(P)=\rho$ and $\Gamma(Q)=\sigma$.
\end{lemma}

In fact, for the spectral decomposition $\sigma^{-1/2}\rho\sigma^{-1/2}=\sum_{x\in\mathcal{X}}\lambda_x\Pi_x$ the infimum is achieved as (see \cite[Sec.~4.2.3]{mybook})
\begin{align}
&\Gamma(\cdot)=\sum_{x\in\mathcal{X}}\frac{\bra{x}\cdot\ket{x}}{Q(x)}\sigma^{1/2}\Pi_x\sigma^{1/2}\,,\\
&P(x)=\lambda_xQ(x),\quad Q(x)=\text{Tr}[\sigma\Pi_x]\,.
\end{align}

Additionally, for the axiomatic extension of the classical R\'enyi divergences from Eq.~\eqref{eq:renyi-classical}, the smallest quantum divergences that are also
\begin{enumerate}
\setcounter{enumi}{2}
\item additive on product states\footnote{Divergences with properties 1)\,--\,3) are also termed relative entropies \cite{gour20}, also see \cite{Khatri20} for more details.}
\end{enumerate}
for $\alpha\in[1/2,\infty]$ are the sandwiched R\'enyi divergences \cite{lennert13,wilde13}
\begin{align}\label{eq:sandwiched}
\tD_{\alpha}(\rho\|\sigma)&=\frac{1}{\alpha-1}\log\tQ_{\alpha}(\rho\|\sigma)\quad\text{with}\nonumber\\
\tQ_{\alpha}(\rho\|\sigma)&=\text{Tr}\left[\left(\sigma^{\frac{1-\alpha}{2\alpha}}\rho\sigma^{\frac{1-\alpha}{2\alpha}}\right)^{\alpha}\right]\;\text{for $\rho\ll\sigma$}
\end{align}
and the respective limits for $\alpha\to1$ and $\alpha\to\infty$ as
\begin{align}
\tD_1(\rho\|\sigma)&=\mathrm{Tr}\left[\rho\left(\log\rho-\log\sigma\right)\right]\quad\text{and}\\
\tD_\infty(\rho\|\sigma)&=\inf\left\{\lambda\in\mathbb{R}\middle|\rho\preceq2^\lambda\sigma\right\}\,,
\end{align}
respectively. We note that $\tD_1(\rho\|\sigma)\neq\widehat{D}_1(\rho\|\sigma)$ unless $\rho$ and $\sigma$ commute. Moreover, the formal definition from Eq.~\eqref{eq:sandwiched} can be extended to $\alpha>0$, but as showcased in \cite[Theorem 7]{Berta2017} only for $\alpha\in[1/2,\infty]$ one has the data-processing inequality from Eq.~\eqref{eq:data-processing}.

We record the following noteworthy property.

\begin{lemma}\label{lem:q_chain}
Let $\rho, \sigma \in \cS(\cH)$, $\cE,\cF$ be positive maps from $\cB(\cH)$ to $\cB(\cH')$, and $\alpha\in(0,\infty]$. Then, we have
\begin{align}\label{eq:1}
\tQ_{\alpha}(\cE(\rho) \| \cF(\sigma)) \leq \big|\spec(\sigma)\big|^{\alpha} \tQ_{\alpha} \big(\cE(\cP_{\sigma}(\rho)) \big\|  \cF(\sigma) \big)\,.
\end{align}
\end{lemma}

\begin{proof}
We follow \cite[Eq.~(4.52)]{mybook}. From the pinching inequality in Eq.~\eqref{eq:pinching} together with the positivity of $\cE$ we get
\begin{align}
\cE(\rho)\preceq\big|\text{spec}(\sigma)\big|\cE(\cP_\sigma(\rho))\,.
\end{align}
Multiplying this from the left and the right with the positive operator $\cF(\sigma)^{\frac{1-\alpha}{2\alpha}}$ leads to
\begin{align}
&\cF(\sigma)^{\frac{1-\alpha}{2\alpha}}\cE(\rho)\cF(\sigma)^{\frac{1-\alpha}{2\alpha}}\nonumber\\
&\preceq\big|\text{spec}(\sigma)\big|\cF(\sigma)^{\frac{1-\alpha}{2\alpha}}\cE(\cP_\sigma(\rho))\cF(\sigma)^{\frac{1-\alpha}{2\alpha}}\,.
\end{align}
The claim then follows by the monotonicity of the trace of the monotone function $t\mapsto t^{\alpha}$.
\end{proof}

Finally, other definitions of interest include the Petz R\'enyi divergences \cite{petz86}, the $\alpha$-$z$-divergences \cite{audenaert13}, and divergences based on convex optimization methods \cite{fawzi21,brown21}. We do not directly employ these further definitions in our work, but note that \cite[Corollary 5.2]{fawzi21} proves quantum extensions of the chain rule from Eq.~\eqref{eq:classical_chain} in terms of divergences based on convex optimization methods.

All of aforementioned quantum R\'enyi divergences are extended to quantum channels $\cE,\cF$ as \cite{Cooney16,leditzky18}
\begin{align}
\mathbb{D}_{\alpha}\big(\cE\big\| \cF\big)=\sup_{\rho\in\cS(\cH)}\mathbb{D}_{\alpha}\big(\cE(\rho)\big\| \cF(\rho)\big)
\end{align}
or alternatively as the (generally larger) stabilized version \cite{wilde20,berta19}
\begin{align}
&\mathbb{D}_{\alpha}^{\text{stab}}\big(\cE\big\| \cF\big)\nonumber\\
&=\sup_{\rho\in\cS(\cH\otimes\cH)}\mathbb{D}_{\alpha}\big((\cE_A\otimes\mathcal{I}_R)(\rho_{AR})\big\| (\cF_A\otimes\mathcal{I}_R)(\rho_{AR})\big)\,,
\end{align}
where $\mathcal{I}_R$ denotes the identity channel on $\cS(\cH_R)$, and the supremum is over all $\rho_{AR}\in\cS(\cH_A\otimes\cH_R)$ and finite dimensional $\cH_R$.\footnote{One might choose $|R|=|A|$ without loss of generality as, e.g., discussed after \cite[Definition II.2]{leditzky18}.} In the general quantum case, the divergences $\mathbb{D}_{\alpha}\big(\cE\big\| \cF\big)$ and $\mathbb{D}_{\alpha}^{\text{stab}}\big(\cE\big\| \cF\big)$ are not additive on product channels (see, e.g., \cite{hastings09,kun20}), and thus one defines the corresponding regularized versions
\begin{align}
\mathbb{D}_{\alpha}^\infty\big(\cE\big\| \cF\big)&=\lim_{n\to\infty}\frac{1}{n}\mathbb{D}_{\alpha}\big(\cE^{\otimes n}\big\| \cF^{\otimes n}\big)\\
\mathbb{D}_{\alpha}^{\text{stab},\infty} \big(\cE\big\| \cF\big)&=\lim_{n\to\infty}\frac{1}{n}\mathbb{D}_{\alpha}^{\text{stab}}\big(\cE^{\otimes n}\big\| \cF^{\otimes n}\big)\,.
\end{align}
Note that these limits are well-defined whenever $\widetilde{D}_{\infty}(\mathcal{E}\|\mathcal{F})$ is finite. To see this consider the sequences $f_n = \frac{1}{n}\mathbb{D}_{\alpha}\big(\cE^{\otimes n}\big\| \cF^{\otimes n}\big)$ where $n \in \mathbb{N}$. We first note that $f_n$ is bounded from above by $\widetilde{D}_{\infty}(\mathcal{E}\|\mathcal{F})$ for all $n$ since $\mathbb{D}_{\alpha}(\rho\|\sigma) \leq \widehat{D}_{\alpha}(\rho\|\sigma) \leq \widetilde{D}_{\infty}(\rho\|\sigma)$ for any pair of states (see, e.g.,~\cite{mybook}) and, furthermore, $\widetilde{D}_{\infty}(\mathcal{E}\|\mathcal{F})$ is additive for tensor product channels. We further get $f_{n+1} \geq \frac{n}{n+1} f_n$ by choosing product inputs and hence the function is monotonic asymptotically and there can only be a single accumulation point.


\section{Main results}\label{sec:main}

Our main result is the following R\'enyi divergences based simple meta chain rule for quantum channels.

\begin{proposition} \label{prop:chain}
Let $\rho, \sigma \in \cS(\cH)$, $\cE,\cF$ be positive maps from $\cB(\cH)$ to $\cB(\cH')$, and $\alpha\in(0,\infty]$. Then, we have
\begin{align}
\mathbb{D}_{\alpha}\big( \mathcal{E}(\rho)  \big\|  \mathcal{F}(\sigma) \big) \leq \widehat{D}_{\alpha}(\rho\|\sigma) +\mathbb{D}_{\alpha} \big(\cE\big\| \cF\big)\,,
\end{align}
where $\mathbb{D}_{\alpha}$ denotes any R\'enyi divergence of order $\alpha$.
\end{proposition}

\begin{proof}
Using Matsumoto's construction (Lemma \ref{lem:matsumoto}), we write
\begin{align}
\mathcal{E}(\rho) &= \sum_x \hat{p}_x\, \mathcal{E}\big(\underbrace{\Gamma\big(\proj{x}\big)}_{=:\rho^x}\big) \quad \textrm{and}\\
\mathcal{F}(\sigma) &= \sum_x \hat{q}_x\, \mathcal{F}\big(\underbrace{\Gamma\big(\proj{x}\big)}_{=\rho^x} \big)\,.
\end{align}
We then use the data-processing inequality to get
\begin{align}
&\mathbb{D}_{\alpha}\big( \mathcal{E}(\rho) \big\|  \mathcal{F}(\sigma) \big)\nonumber\\
&\leq 
\mathbb{D}_{\alpha} \left( \sum_x \hat{p}_x \proj{x} \otimes \mathcal{E}\left(\rho^x\right) \middle\|  \sum_x \hat{q}_x \proj{x} \otimes \mathcal{F}\left(\rho^x\right)  \right) \\
&=\frac{1}{\alpha-1}\log\left(\sum_x\hat{p}_x^{\alpha}\hat{q}_x^{1-\alpha}2^{(\alpha-1)\mathbb{D}_{\alpha}\left(\mathcal{E}\left(\rho^x\right)\middle\|\mathcal{F}\left(\rho^x\right)\right)}\right)\\
&\leq D_\alpha(\hat{p} \| \hat{q} ) + \max_x \mathbb{D}_{\alpha} \left( \mathcal{E}\left(\rho^x\right) \middle\| \mathcal{F}\left(\rho^x\right)  \right)\,,
\end{align}
which implies the statement we set out to prove.
\end{proof}

The corresponding geometric chain rules for $\alpha\in(0,2]$ as
\begin{align}\label{eq:geometric-chain}
\widehat{D}_{\alpha}\big( \mathcal{E}(\rho)  \big\|  \mathcal{F}(\sigma) \big) \leq \widehat{D}_{\alpha}(\rho\|\sigma) +\widehat{D}_{\alpha} \big(\cE\big\| \cF\big)\,,
\end{align}
generalize the stabilized geometric chain rules from \cite[Lemma 7]{fang21} that were proven with different techniques for completely positive maps $\cE_A,\cF_A$ and $\alpha\in(1,2]$ as
\begin{align}\label{eq:geometric-chain2}
&\widehat{D}_{\alpha}\big( (\mathcal{E}_A\otimes\mathcal{I}_R)(\rho_{AR})  \big\|  (\mathcal{F}_A\otimes\mathcal{I}_R)(\sigma_{AR}) \big)\nonumber\\
&\qquad \leq \widehat{D}_{\alpha}(\rho_{AR}\|\sigma_{AR}) +\widehat{D}_{\alpha}^{\text{stab}} \big(\cE_A\big\| \cF_A\big)\,,
\end{align}
as well as their extension to $\alpha\in(0,1)$ \cite[Proposition 45]{Katariya21}. We note that the stabilized version of the channel divergence was required in~\eqref{eq:geometric-chain2}, whereas our Proposition~\ref{prop:chain} does not require stabilization.

To obtain chain rules without featuring geometric R\'enyi divergences, we can apply pre-processing maps in terms of rank-one projective measurements $\Lambda(\cdot)$. Namely, for $\rho,\sigma\in\cS(\cH)$ we have that the pre-processed inputs $\Lambda(\rho),\Lambda(\sigma)$ commute and thus we find the chain rule
\begin{align}\label{eq:pre-processing}
\mathbb{D}_{\alpha}\big( \mathcal{E}( \Lambda(\rho))  \big\|  \mathcal{F}(\Lambda(\sigma)) \big) 
\leq D_{\alpha,M}(\rho\|\sigma) + \mathbb{D}_{\alpha} \big(\cE\big\| \cF \big) \,.
\end{align}
We can then again eliminate the pre-processing map for the sandwiched R\'enyi divergences with $\alpha\in(1,\infty]$ by choosing the overall rank-one projective measurement map $\Lambda(\cdot)=\mathcal{M}_\sigma(P_\sigma(\cdot))$ with the pinching map $P_\sigma(\cdot)$ and the projective rank-one measurement map $\mathcal{M}_\sigma(\cdot)$ in the joint eigenbasis of $\sigma$ and $P_\sigma(\rho)$\,---\,then leaving those states invariant. Applying Lemma \ref{lem:q_chain} for $\alpha\in(1,\infty]$ immediately gives
\begin{align}
\tD_{\alpha}\big( \mathcal{E}( \Lambda(\rho))  \big\|  \mathcal{F}(\Lambda(\sigma)) \big)\geq&\;\tD_{\alpha}\big( \mathcal{E}(\rho)  \big\|  \mathcal{F}(\sigma) \big)\nonumber\\
&-\frac{\alpha}{\alpha-1}\log\big|\text{spec}(\sigma)\big|
\end{align}
and thus we find the following sandwiched chain rule.

\begin{corollary}\label{cor:sandwiched-chain}
Let $\rho, \sigma \in \cS(\cH)$, $\cE,\cF$ be positive maps from $\cB(\cH)$ to $\cB(\cH')$, and $\alpha\in(1,\infty]$. Then, we have
\begin{align}\label{eq:sandwiched-chain}
\tD_{\alpha}\big(\cE(\rho) \big\| \cF(\sigma) \big)\leq&\;D_{\alpha,M}(\rho\|\sigma) +\tD_{\alpha}\big(\cE\big\| \cF \big)\nonumber\\
&+ \frac{\alpha}{\alpha-1} \log \big|\spec(\sigma)\big| \,.
\end{align}
\end{corollary}

As a special case, applying above corollary to unital maps $\cF$, the identity matrix $\sigma=1$,\footnote{The definition for $\mathbb{D}_{\alpha}(\rho\|\sigma)$ and its properties straightforwardly apply to general positive operators $\sigma$ on $\cH$ as well.} and $\rho\in\cS(\cH)$, leads for $\alpha\in[1,\infty]$ to
\begin{align}\label{eq:renyi-unital}
H_{\alpha}(\cE(\rho))-H_{\alpha}(\rho)\geq-\inf_{\cF:\,\text{unital}}\tD_{\alpha}(\cE\|\cF)\,,
\end{align}
featuring the quantum R\'enyi entropies $H_\alpha(\rho)=\frac{1}{1-\alpha} \log \mathrm{Tr}\left[\rho^{\alpha}\right]$ and in particular the von Neumann entropy $H_1(\rho)=-\mathrm{Tr}\left[\rho\log\rho\right]$ in the limit $\alpha\to1$. This generalizes the well-known monotonicity property $H_{\alpha}(\cE(\rho))-H_{\alpha}(\rho)\geq0$ for unital quantum channels $\cE$ to the Eq.~\eqref{eq:renyi-unital} for generic positive maps.

Going back to the general case, but restricting our attention to tensor-stable positive maps $\mathcal{E}$ and $\mathcal{F}$, such that $\mathcal{E}^{\otimes n}$ and $\mathcal{F}^{\otimes n}$ are positive maps. Applying Corollary \ref{cor:sandwiched-chain} to $\rho^{\otimes n}$, $\sigma^{\otimes n}$, $\mathcal{E}^{\otimes n}$ and $\mathcal{F}^{\otimes n}$, and using that $\left|\spec\left(\sigma^{\otimes n}\right)\right|\leq\text{poly}(n)$, we arrive at the corresponding regularized sandwiched chain rule.

\begin{corollary}\label{cor:alpha}
Let $\rho, \sigma \in \cS(\cH)$, $\cE,\cF$ be tensor-stable positive maps from $\cB(\cH)$ to $\cB(\cH')$, and $\alpha\in(1,\infty]$. Then, we have
\begin{align}
\tD_{\alpha}\big(\cE(\rho) \big\| \cF(\sigma) \big)\leq\tD_{\alpha}(\rho\|\sigma)+\tD_{\alpha}^\infty\big(\cE\big\| \cF \big)\,.
\end{align}
\end{corollary}

This can be compared to \cite[Theorem III.1]{christandl17} of which a special case states for the same range $\alpha\in(1,\infty]$ as in Corollary \ref{cor:alpha}\,---\,but for positive trace preserving maps instead of tensor-stable maps\,---\,that
\begin{align}
\tD_{\alpha}\big(\cE(\rho) \big\| \cF(\sigma) \big)\leq\tD_{\alpha}(\rho\|\sigma) +\tD_\infty\big(\cE\big\| \cF \big)\,.
\end{align}
Applied to quantum channels, this is weaker than Corollary \ref{cor:alpha} as the sandwiched R\'enyi divergences are monotone in $\alpha\in[1/2,\infty]$ and one has $\tD_\infty^\infty\big(\cE\big\| \cF \big)=\tD_\infty\big(\cE\big\| \cF \big)$ \cite[Proposition 10]{wilde20}. Corollary \ref{cor:alpha} also generalizes the stabilized sandwiched chain rules from \cite[Corollary 5.2]{fawzi21} that hold for completely positive maps $\cE_A,\cF_A$, $|R|=|A|$, and $\alpha\in(1,\infty]$ as
\begin{align}
&\tD_{\alpha}\big((\cE_A\otimes\mathcal{I}_R)(\rho_{AR}) \big\| (\cF_A\otimes\mathcal{I}_R)(\sigma_{AR}) \big)\nonumber\\
&\qquad \leq
\tD_{\alpha}(\rho_{AR}\|\sigma_{AR}) +\tD_{\alpha}^{\text{stab},\infty}\big(\cE_A\big\| \cF_A \big)
\end{align}
and were proven with different techniques.

Note that for $\alpha\in(0,1]$ it is unfortunately unclear how to eliminate the pre-processing map and directly regularizing Eq.~\eqref{eq:pre-processing} for general $\alpha\in(0,\infty]$ just leads to
\begin{align}\label{eq:preproc-conjecture}
&\limsup_{n \to \infty}\sup_{\Lambda^n}\frac{1}{n}\mathbb{D}_{\alpha} \left( \mathcal{E}^{\otimes n} \left( \Lambda^n\left(\rho^{\otimes n} \right)\right)  \big\|  \mathcal{F}^{\otimes n} \left(\Lambda^n\left(\sigma^{\otimes n}\right)\right) \right)\nonumber\\
&\qquad \leq \tD_{\alpha}(\rho\|\sigma) + \mathbb{D}_{\alpha}^\infty(\cE \big\| \cF ) \,,
\end{align}
where the supremum is over all rank-one projective measurement maps $\Lambda^n(\cdot)$ acting on $n$ copies. It is then only for $\alpha\in(1,\infty]$ that the left-hand side of Eq.~\eqref{eq:preproc-conjecture} is further lower bounded by $\tD_{\alpha}\big(\cE(\rho) \big\| \cF(\sigma) \big)$ (Lemma \ref{lem:q_chain}), whereas this remains unclear for $\alpha\in(0,1]$.\footnote{In this context, Lemma \ref{lem:q_chain} is not useful for $\alpha\in(0,1]$ as one would have to divide the logarithm on both sides of Eq.~\eqref{eq:1} with the pre-factor $\alpha-1$, which becomes negative in this range.} Notwithstanding, at least for $\alpha=1$ it is known that \cite[Theorem 3.5]{kun20}
\begin{align}
\tD_1\big(\cE(\rho) \big\| \cF(\sigma) \big)\leq\tD_1(\rho\|\sigma)+\tD_1^\infty\big(\cE\big\| \cF \big)\,.
\end{align}
However, this does not directly follow from our R\'enyi divergences chain rules, as the order of taking the limits in $n\to\infty$ and $\alpha\to1$ might ultimately matter when regularizing Eq.~\eqref{eq:sandwiched-chain}. We briefly get back to this question in Section \ref{sec:outlook} and present some conjectures.


\section{Applications}\label{sec:applications}

R\'enyi divergences chain rules such as Eq.~\eqref{eq:classical_chain} and its quantum extensions allow to characterize the asymptotic error exponents of adaptive channel discrimination \cite{hayashi09d,polyanski11,berta19,wilde20,kun20,salek20,fang21,fawzi21}. Namely, the crucial quantities to study are the so-called amortized channel divergences \cite{wilde20,berta19,Wang19-2}
\begin{align}
\mathbb{D}_{\alpha}^{\text{a}} \big(\cE\big\| \cF\big)=\sup_{\rho,\sigma}\Big\{&\mathbb{D}_{\alpha} \big((\cE_A\otimes\mathcal{I})(\rho_{AR})\big\|(\cF_A\otimes\mathcal{I}_R)(\sigma_{AR})\big)\nonumber\\
&-\mathbb{D}_{\alpha}\big(\rho_{AR}\big\|\sigma_{AR}\big)\Big\}\,,
\end{align}
where the supremum is over all $\rho_{AR},\sigma_{AR}\in\cS(\cH_A\otimes\cH_R)$ and finite dimensional $\cH_R$.\footnote{No bound on the dimension $|R|$ is known.} For classical \cite{hayashi09d,polyanski11} and classical-quantum channels \cite{berta19,wilde20,salek20} it is derived from the chain rules in Eq.~\eqref{eq:classical_chain} and its classical-quantum extensions that for $\alpha>0$ and all the relevant channel divergences
\begin{align}
\mathbb{D}_{\alpha}^{\text{a}} \big(\cE\big\| \cF\big)=\mathbb{D}_{\alpha} \big(\cE\big\| \cF\big)\,.
\end{align}
This then means that adaptive or correlated input strategies are of no asymptotic advantage compared to product strategies and one immediately obtains a full asymptotic characterization of channel discrimination (see aforementioned references for details). In the general quantum case, one only has the stabilized and regularized chain rules as in Corollary \ref{cor:alpha}, and for $\alpha\in[1,\infty]$ one then derives \cite{kun20,fawzi21}
\begin{align}
\tD_{\alpha}^{\text{a}} \big(\cE\big\| \cF\big)=\tD_{\alpha}^{\text{stab},\infty} \big(\cE\big\| \cF\big)\,,
\end{align}
which includes $\alpha=1$ thanks to \cite[Theorem 3.5]{kun20}. Consequently, this shows that adaptive strategies are of no asymptotic advantage over entangled input strategies\,---\,with the regularization reflecting that the latter are still needed compared to product input strategies \cite{kun20}. Nonetheless, for a full asymptotic characterization of channel discrimination technical question concerning the $\alpha\to1$ limit remain \cite{fawzi21} (cf.~the discussion in Section \ref{sec:outlook}). Alternatively, one might resort to single-letter converse bounds as provided by the geometric chain rules Eqs.~\eqref{eq:geometric-chain}--\eqref{eq:geometric-chain2} giving that
\begin{align}
\widehat{D}_{\alpha}^{\text{a}} \big(\cE\big\| \cF\big)=\widehat{D}_{\alpha}^{\text{stab}} \big(\cE\big\| \cF\big)\,.
\end{align}
We refer to \cite{fang21} for details. One advantage, of our novel geometric chain rule in Eq.~\eqref{eq:geometric-chain} is that it does not require a stabilization system and hence one might explore channel discrimination tasks without additional quantum memory as done in the classical case \cite{hayashi09d} (or bounded quantum memory of fixed dimension).


\section{Outlook}\label{sec:outlook}

We hope that our matrix analysis based techniques will be more broadly useful to derive entropy inequalities for quantum channels. Many fundamental questions remain with respect to the goal of giving for $\alpha>0$ the most general R\'enyi divergence chain rules for quantum channels. For example, our main results from Section \ref{sec:main} immediately raise the question if
\begin{align}\label{eq:main-conjecture}
&\limsup_{n \to \infty}\sup_{\Lambda^n}\frac{1}{n}\tD_1 \left( \mathcal{E}^{\otimes n} \left( \Lambda^n\left(\rho^{\otimes n} \right)\right)  \big\|  \mathcal{F}^{\otimes n} \left(\Lambda^n\left(\sigma^{\otimes n}\right)\right) \right)\nonumber\\
&\qquad \geq\tD_1\big(\cE(\rho) \big\| \cF(\sigma) \big)\,,
\end{align}
where the supremum is over all rank-one projective measurement maps $\Lambda^n(\cdot)$ acting on $n$ copies. Choosing the standard pinching based measurement map $\Lambda^n(\cdot)=\mathcal{M}_{\sigma^{\otimes n}}(P_{\sigma^{\otimes n}}(\cdot))$ in the form of Eq.~\eqref{eq:pinching-measure} only leads to
\begin{align}
&\frac{1}{n} \sup_{\Lambda^n}  \tD_1 \left( \mathcal{E}^{\otimes n} \left( \Lambda^n\left(\rho^{\otimes n} \right)\right)  \big\|  \mathcal{F}^{\otimes n} \left(\Lambda^n\left(\sigma^{\otimes n}\right)\right) \right)\nonumber\\
&\geq\frac{1}{n} \tD_1 \left( \mathcal{E}^{\otimes n} \left(\int\mu(\mathrm{d}t)\,\left(\sigma^{it}\rho\sigma^{-it}\right)^{\otimes n}\right)  \big\|  \left(\mathcal{F}\left(\sigma\right)\right)^{\otimes n} \right)\\
&\geq\min_t\tD_1 \left(\mathcal{E}\left(\sigma^{it}\rho\sigma^{-it}\right) \middle\|\mathcal{F}(\sigma)\right)-\frac{1}{n}\log\text{poly}(n)\,,
\end{align}
where we used in the last inequality that by Carath\'eodory's theorem the integral $\int\mu(\mathrm{d}t)(\cdot)^{\otimes n}$ is discretized with at most $\text{poly}(n)$ terms \cite[Theorem D.5]{bertachristandl11} and can then by \cite[Lemma 2.5]{berta21} be taken outside of the divergence at the price of the asymptotically small $\log\text{poly}(n)$ fudge term. The appearance of unitary operators like $\sigma^{it}$ is reminiscent of previous work on quantum entropy inequalities \cite{fawzirenner14,brandao14,wilde15,sutter16a,junge18,sutter16,sutter17,berta21}, but here we would need to eliminate them.

As the $\alpha\in(1,\infty]$ analogue of Eq.~\eqref{eq:main-conjecture} holds by Corollary \ref{cor:sandwiched-chain}, we might also ask more generally what happens for the whole parameter range $\alpha\in(0,1]$. This is strongly connected with the question if $\widehat{D}_{\alpha}^{\text{stab},\infty} \big(\cE\big\| \cF\big)$ is continuous for $\alpha\to1$ \cite{fawzi21}. Any insight on this question would have widespread applications in quantum information theory \cite{Fang22}\,---\,starting from the fundamental statistical task of channel discrimination.

{\it Note added:} After submitting our paper, the work \cite{Metger22} on generalised entropy accumulation appeared. A crucial proof step is their complementary R\'enyi divergence chain rule \cite[Theorem 3.1]{Metger22} for quantum channels and it would be interesting to explore the interplay with our purely matrix analysis based techniques.


\paragraph*{Acknowledgements} This work was completed prior to MB joining the AWS Center for Quantum Computing. This research is supported by the National Research Foundation, Prime Minister’s Office, Singapore and the Ministry of Education, Singapore under the Research Centres of Excellence programme. MT is also supported in part by NUS startup grants (R-263-000-E32-133 and R-263-000-E32-731). We thank David Sutter for discussions.


\appendix

We mention the following, alternative proof of the sandwiched chain rule (Corollary \ref{cor:sandwiched-chain}). This directly lifts the classical chain rule from Eq.~\eqref{eq:classical_chain} to the quantum case.

\begin{proof}[Proof of Corollary \ref{cor:sandwiched-chain}]
Let $\{ \ket{x} \}_x$ be a joint eigenbasis for both $\sigma$ and $\cP_{\sigma}(\rho)$. We write
\begin{align}
\cE(\cP_{\sigma}(\rho)) &= \sum_{x} \underbrace{\langle x | \rho | x \rangle}_{=:p_x} \underbrace{ \cE\big( \proj{x} \big) }_{=:\rho^x}
\quad \textrm{and}\\
\cF(\sigma) &= \sum_{x} \underbrace{\langle x | \sigma | x \rangle}_{=:q_x} \underbrace{ \cF\big( \proj{x} \big) }_{=:\sigma^x}
\end{align}
and by the data-processing for the partial trace we find
\begin{align}
&\tQ_{\alpha} \big(\cE(\cP_{\sigma}(\rho)) \big\|  \cF(\sigma) \big)\\
&\leq \tQ_{\alpha}\left( \sum_{x} {p_x} \proj{x} \otimes \rho^x \,\middle\|\, \sum_{x} {q_x} \proj{x} \otimes \sigma^x \right) \\
&= \sum_x p_x^{\alpha} q_x^{1-\alpha} \tQ_{\alpha}\left( \rho^x \middle\| \sigma^x \right) \\
&\leq \sum_x p_x^{\alpha} q_x^{1-\alpha} \cdot \max_x \tQ_{\alpha}\left( \rho^x \middle\| \sigma^x \right) \\
&\leq 2^{(\alpha-1)D_{\alpha,M}(\rho\|\sigma)} \cdot \max_{\omega \in \cS} \tQ_{\alpha} \big(\cE(\omega) \big\| \cF(\omega) \big) \,,
\end{align}
where in the last inequality we used that $\rho^x$ and $\sigma^x$ are channel outputs for the same input. Combining this with Lemma \ref{lem:q_chain} and taking the logarithm as well as multiplying by $\frac{1}{\alpha-1}$, yields the desired result.
\end{proof}




\bibliographystyle{plain}
\bibliography{library_MT}

\end{document}